\documentclass[11pt]{article}
\usepackage[margin=1.125in]{geometry}
\usepackage{xcolor}
\definecolor{DarkGreen}{rgb}{0.1,0.5,0.1}
\definecolor{DarkRed}{rgb}{0.5,0.1,0.1}
\definecolor{DarkBlue}{rgb}{0.1,0.1,0.5}

\usepackage[small]{caption}
\usepackage[pdftex]{hyperref}
\hypersetup{
    unicode=false,          % non-Latin characters in Acrobat¿s bookmarks
    pdftoolbar=true,        % show Acrobat toolbar?
    pdfmenubar=true,        % show Acrobat menu?
    pdffitwindow=false,      % page fit to window when opened
    pdfnewwindow=true,      % links in new window
    colorlinks=true,       % false: boxed links; true: colored links
    linkcolor=DarkBlue,          % color of internal links
    citecolor=DarkGreen,        % color of links to bibliography
    filecolor=DarkGreen,      % color of file links
    urlcolor=DarkBlue,          % color of external links
    pdftitle={},
    pdfauthor={},    
}
\usepackage{cite}
\usepackage{amssymb,amsmath,amsthm,amsfonts,enumitem}
\usepackage{fullpage,nicefrac,comment}
\usepackage{tikz}
\usetikzlibrary{arrows,decorations.pathmorphing,decorations.shapes,snakes,patterns,positioning}
\usepackage[ruled,linesnumbered]{algorithm2e}

\newcommand{\FC}{FindList} % This is the name of the algorithm that does linear algebra
\newcommand{\LDS}{SlowListDecode} % This is the name of the proto-list-decoding algorithm that runs too slowly

\newcommand{\cC}{\ensuremath{\mathcal{C}}}

\newcommand{\F}{{\mathbb F}}

\newcommand{\inset}[1]{\left\{#1\right\}}

\newcommand{\inparen}[1]{\left(#1\right)}

\newcommand{\suchthat}{\,:\,}

\newcommand{\supp}{\mathrm{Supp}}

\newcommand{\poly}{\mathrm{poly}}

\newcommand{\Ker}{\ensuremath{\operatorname{Ker}}}

\newcommand{\eps}{\varepsilon}
\renewcommand{\epsilon}{\varepsilon}

\newtheorem{theorem}{Theorem}[section] 
\newtheorem*{theorem*}{Theorem} 
\newtheorem{lemma}[theorem]{Lemma} 
\newtheorem{definition}[theorem]{Definition}

\newtheorem{observation}[theorem]{Observation}

\newtheorem{claim}[theorem]{Claim}

\newcommand{\restatethm}[3]{\medskip \noindent {\bf #1} (#2). \textit{#3}\\}

\title{Linear-time Erasure List-decoding of Expander Codes}
\author{
Noga Ron-Zewi\thanks{Department of Computer Science, University of Haifa. \texttt{noga@cs.haifa.ac.il}. Research supported in part by BSF grant 2017732} 
\and 
Mary Wootters\thanks{Department of Computer Science and Department of Electrical Engineering, Stanford University.  \texttt{marykw@stanford.edu}.  Research supported in part by NSF CAREER award CCF-1844628 and by NSF-BSF award CCF-1814629, as well as by a Sloan Research Fellowship.}
\and 
Gilles Z\'{e}mor\thanks{Institut de Math\'{e}matiques de Bordeaux, Universit\'{e} de Bordeaux.  \texttt{zemor@math.u-bordeaux.fr}.  
}}
\date{\today}
\begin{document}
\maketitle

\begin{abstract}
We give a linear-time erasure list-decoding algorithm for expander codes. More precisely, let $r > 0$ be any integer.  Given an inner code $\cC_0$ of length $d$, and a $d$-regular bipartite expander graph $G$ with $n$ vertices on each side, we give an algorithm to list-decode the expander code $\cC = \cC(G, \cC_0)$ of length $nd$ from approximately $\delta \delta_r nd$ erasures in time $n \cdot \poly(d2^r / \delta)$, where $\delta$ and $\delta_r$ are the relative distance and the $r$'th generalized relative distance of $\cC_0$, respectively. To the best of our knowledge, this is the first linear-time algorithm that can list-decode expander codes from erasures beyond their (designed) distance of approximately $\delta^2 nd$.

To obtain our results, we show that an approach similar to that of (Hemenway and Wootters, \emph {Information and Computation}, 2018) can be used to obtain such an erasure-list-decoding algorithm with an exponentially worse dependence of the running time on $r$ and $\delta$; then  we show how to improve the dependence of the running time on these parameters.
\end{abstract}

\section{Introduction}\label{sec:intro}

In coding theory, the problem of \em list-decoding \em is to return all
codewords that are close to some received word $z$; in \em algorithmic
list-decoding, \em the problem is to do so efficiently.  While there has been a
great deal of progress on algorithmic list-decoding in the past two
decades~\cite{GS-list-dec,PV05,GR06a,GW17,GX12,GX13,Kop15,GK14,HRW19,KRSW18},
most work has relied crucially on algebraic constructions, and thus it is interesting to develop combinatorial tools to construct efficiently list-decodable codes with good parameters. 

In this work, we consider the question of list-decoding \em expander codes, \em introduced by Sipser and Spielman in \cite{SS96}.
We define expander codes formally in Section~\ref{sec:prelim}, but briefly, the expander code $\cC(G, \cC_0)$ is a linear code constructed from a $d$-regular bipartite expander graph $G$ and a linear inner code $\cC_0 \subseteq \F_2^d$.  A codeword of $\cC(G, \cC_0)$ is a vector in $\F_2^{E(G)}$ which is a labeling of edges in $G$.  The constraints are that, for each vertex $v$ of $G$, the labels on the $d$ edges incident to $v$ form a codeword in $\cC_0$.

Expander codes are notable for their very efficient unique
decoding algorithms~\cite{SS96, Z01,LMSS01,SR03,AS05,BZ02,BZ05,BZ06, RS06, HOW13}.
However, very little is known about the algorithmic list-decodability of expander codes, and it is an open problem to find a family of expander codes that admit fast linear-time list-decoding algorithms with good parameters.
Motivated by this open problem, our main contribution is a linear-time algorithm for list decoding expander codes \em from erasures. \em  

\paragraph{Erasure list-decoding.} Erasure-list-decoding is a variant of list-decoding where the received word $z$ may have some symbols which are ``$\bot$'' (erasures), and the goal is to recover all codewords consistent with~$z$.  More formally, let $\cC \subseteq \F_2^N$ be a binary code of length $N$. 
For $z \in (\F_2 \cup \bot)^N$, define
\[\mathrm{List}_{\cC}(z) := \left\{ c \in \cC \suchthat c_i = z_i \text{ whenever } z_i \neq \bot \right\}. \]
We say that $\cC$ is \em erasure-list-decodable \em from $e$ erasures with list size $L$ if for any $z \in (\F_2 \cup \{\bot\})^N$ with at most $e$ symbols equal to $\bot$, $|\mathrm{List}_{\cC}(z)| \leq L$.

Erasure list-decoding has been studied before~\cite{Gur03,GI02,GI04a,GO06,DJX14,HW18,BDT18}, motivated both by standard list-decoding and as an interesting combinatorial and algorithmic question in its own right.
It is known that the erasure-list-decodability of a linear code is precisely captured by its generalized distances.  The \em $r$'th (relative)\footnote{Throughout this paper, we will work with the relative generalized distances (that is, measured as a fraction of coordinates).  We will omit the adjective ``relative'' to describe these quantities in the future.}  generalized  distance \em $\delta_r$ of a linear code $\cC \subseteq \F_2^N$ is the minimum fraction of coordinates   which are not identically zero in an $r$-dimensional subspace of $\cC$, that is,
\[ \delta_r = \frac{1}{N} \min_{V} \big|\inset{ i \suchthat \exists v\in V, v_i \neq 0 }\big|,\]
where the minimum is taken over all linear subspaces $V \subseteq \cC$ of dimension $r$.
Thus, $\delta_1$ coincides with the traditional \em (relative) distance \em $\delta$ of the code, which for linear codes equals the minimum relative weight of any nonzero codeword.  
The generalized distances of a linear code $\cC$ characterize its erasure list-decodability:

\begin{lemma}[\cite{Gur03}]\label{lem:gd}
Let $\cC \subseteq \F_2^N$ be a linear code. Then $\cC$ is erasure-list-decodable from $e$ erasures with list size $L$ if and only if $\delta_r(\cC) > e/N$, where $r = 1 + \lfloor \log_2(L) \rfloor$.
\end{lemma}

If $\cC$ is linear, then it can be erasure list-decoded in polynomial time by solving a linear system.  Thus, the combinatorial result of Lemma~\ref{lem:gd} comes with a polynomial-time algorithm.  

Our goal in this paper is twofold.  First, 
we aim to develop algorithms to erasure list-decode expander codes beyond the minimum distance of the code with small list size.  Second, we aim to do so in linear time, faster than the straightforward algorithm described above.

\paragraph{Our Results.}
Our main result is a linear-time erasure list-decoding algorithm for expander codes beyond the (designed) minimum distance.
\begin{theorem}\label{thm:main}
Let $\cC_0 \subseteq \F_2^d$ be a linear code with distance $\delta$ and $r$'th generalized distance $\delta_r$.  Let $G = (L \cup R,E)$ be the double cover
\footnote{The \em double cover \em of a graph $\tilde{G} = (\tilde{V},\tilde{E})$
is the bipartite graph $G = (L \cup R,E)$ defined as follows.  Let $L$ and $R$ be
two copies of $\tilde{V}$; there is an edge between $u \in L$ and $v \in R$ if and only if $(u,v) \in \tilde{E}$ (see Section \ref{sec:prelim}).}
 of a $d$-regular expander graph on $n$ vertices with expansion $\lambda = \max\{\lambda_2, |\lambda_n|\}$.
Let $\cC = \cC(G, \cC_0)$ be the expander code that results.
Let $\eps > 0$, and suppose that
$ \frac{\lambda}{d} \leq \frac{ \eps^2 \delta^2 }{ 2^{r+4} }. $

Then there is an algorithm \textsc{ListDecode} which,
given a received word $z \in (\F_2 \cup \inset{ \bot })^E$ with at most
$(1 - \eps)\delta \delta_rdn$
erasures,
 runs in time 
$ n \cdot \poly\inparen{ \frac{ 2^r d }{ \eps \delta} },$
and returns a matrix $L\in \F_2^{nd \times a}$ and a vector $\ell \in \F_2^{nd}$ so that 
$ \mathcal{L} := \mathrm{List}_{\cC}(z) = \inset{ Lx + \ell\suchthat x \in \F_2^a } $
where $a := \dim(\mathcal{L})$ satisfies
$ a \leq \frac{ 2^{2r+7} }{ \eps^4 \delta^4 }.$
\end{theorem}

Because $\delta_r > \delta$ for any non-trivial linear code (any code of dimension $>1$), the radius that Theorem~\ref{thm:main} achieves is beyond the (designed) minimum distance of $\cC$, which is approximately $\delta^2 d n$.  To the best of our knowledge, this is the first linear-time list-decoding algorithm for expander codes that achieves this with a non-trivial list size.

In light of Lemma~\ref{lem:gd}, the ultimate result we can hope for is 
an algorithm that list-decodes up to $\delta_r(\cC)$ fraction of 
erasures with list size $2^{r-1}$ for any $r \geq 1$. The quantity
$\delta(\cC_0)  \cdot \delta_r(\cC_0)$ in Theorem~\ref{thm:main} may
suggest it plays the role of a 'designed' $r$'th generalized distance,
especially since for $r=1$ it does (up to an $\epsilon$ term) coincide
with the expander designed distance. However, we cannot expect
$\delta(\cC_0)  \cdot \delta_r(\cC_0)$ to be a general lower bound
on the $r$'th generalized distance of an expander code,
which implies in particular that the list-size in Theorem~\ref{thm:main} 
has to be larger than $2^{r}$. Indeed,
already in the special case of \em tensor codes \em (i.e., when the 
graph $G$ is the complete bipartite graph that has perfect expansion), 
the generalized distance has been shown \cite{WY93, Scha00} to be a 
complicated quantity that can be lower than $\delta(\cC_0)  \cdot 
\delta_r(\cC_0)$: in the general expander case, finding a reasonable
description of the worst-case behavior of generalized distances seems
quite challenging.

Note however that our results do imply a weak bound on the generalized distances of an expander code, namely that $\delta_r(\cC)$ is approximately at least $\delta(\cC_0) \cdot \delta_{\Theta(\log r)}( \cC_0)$. Moreover, for the special case of $r=2$, 
we are able to show the following bound on the second generalized distance of an expander code. 

\begin{lemma}\label{lem:delta2} 
Let $\cC_0 \subseteq \F_2^d$ be a linear code with distance $\delta$ and second generalized distance $\delta_2$, and
let $G = (L \cup R,E)$ be the double-cover of a $d$-regular expander graph with expansion $\lambda$.
Let $\epsilon>0$, and suppose that $\frac \lambda d \leq \frac{\delta_2 \delta^2 \epsilon^2} {16}$.
Then the expander code $\cC(G, \cC_0)$ has second generalized distance at least 
$
(1-\epsilon)\cdot \delta \cdot \min\{\delta_2, 2\delta\} .
$
\end{lemma}

Note that under the mild assumption that $\delta_2 ( \cC_0) \leq  2\delta( \cC_0)$ (satisfied by any code that has two minimum-weight codewords with disjoint support),
 the above lemma gives a lower bound of approximately
 $\delta_2( \cC_0) \delta ( \cC_0)$ on the second generalized distance of expander codes.

Finally, note that while we do not know if the list size returned by our algorithm can be generally improved, our algorithm can still list-decode an expander code $\cC$ from up to a $\delta_r(\cC)$ fraction of erasures with list size $2^{r-1}$ for some values of $r$: If $r'$ is such that $\delta_{r'}(\cC) <\delta(\cC_0) \delta_r(\cC_0)$ for some $r = O(1)$, our algorithm will run in linear time and return a list of size $2^{r'-1}$ given a $\delta_{r'}(\cC)$ fraction of erasures.

\subsection{Technical Overview}\label{sec:overview}
In this section, we give a brief overview of our approach.  The basic idea is similar to the approach in \cite{HW18}; however, as we discuss more in Section~\ref{sec:related} below, in that work the goal was \em list-recovery, \em a generalization of list-decoding.  In this work we can do substantially better by restricting our attention to list-decoding, as well as by tightening the analysis of \cite{HW18}.

Let $G=(L \cup R,E)$ be the double-cover of a $d$-regular expander graph, and let $\cC_0 \subseteq \F_2^d$ be a linear code with distance $\delta$ and 
$r$-th generalized distance $\delta_r$.
Since the inner code $\cC_0$ is linear and has $r$'th generalized distance $\delta_r$, there is an $O(d^3)$-time algorithm to erasure list-decode $\cC_0$ from up to $\delta_r d$ erasures.  Our first step will be to do this at every vertex $v \in L \cup R$ that we can, to produce a list $L_v$ at each such vertex.   

In order to ``stitch together'' these lists, we define a notion of \em equivalence \em between edges, similar to the notion in \cite{HW18}.
Suppose that $(u,v)$ and $(w,v)$ are edges incident to a vertex $v$, so that there is some $b \in \F_2$ so that for any $c \in L_v$, $c_{(u,v)} = b + c_{(w,v)}$.  Then, even if we have not pinned down a symbol for $(u,v)$ or $(w,v)$, 
we know that for any legitimate codeword $c \in \mathrm{List}_{\cC}(z)$,
assigning a symbol for one of these edges implies an assignment for the other.
In this case, we say that $(u,v) \sim (w,v)$.  Because the lists $L_v$ are
actually affine subspaces, there are not many equivalence classes at each vertex
(and in particular substantially fewer equivalence classes than in the approach used in \cite{HW18}).

With these equivalence classes defined, we actually give two algorithms, \textsc{\LDS} and \textsc{ListDecode}.   As the name suggests, \textsc{\LDS} is a warm-up that has a worse dependence on $\epsilon, \delta$ and $r$, but is easier to understand.  We describe \textsc{\LDS} (given in Section \ref{sec:slow}, Figure~\ref{fig:alg1}) here first, and then describe the changes that need to be made to arrive at our final algorithm, \textsc{ListDecode} (given in Section \ref{sec:faster}, Figure~\ref{fig:alg2}).

The main idea of \textsc{\LDS} is to choose $s = \poly(2^{r}, 1/\epsilon, 1/\delta)$ large equivalence classes and generate a list of all $2^s$ possible labelings for these equivalence classes.  For each such labeling, we now hope to uniquely fill in the rest of the codeword, to arrive at a list of size $2^s$. 
One might hope that labeling 
these $s$ large equivalence classes would leave a fraction of unlabeled symbols
less than the designed distance of $\cC$,
allowing us to immediately use the known linear-time erasure unique decoding algorithm
for the expander code.  Unfortunately, this is not in general the case.  However, we show that there are many vertices $v$ so that the number of unlabeled edges incident to $v$ is at most $\delta(\cC_0)d$.  Thus, we may run the unique decoder for $\cC_0$ (in time $O(d^3)$) at each such vertex to generate yet more labels.  It turns out that at this point, there \em are \em enough labels to run $\cC$'s unique decoding algorithm and finish off the labeling.

Naively, the algorithm described above runs in time at least $2^s$, since we must loop over all $2^s$ possibilities.  This is exponential in $\eps$ and $\delta$ and doubly-exponential in $r$.  
The idea behind our final algorithm \textsc{ListDecode} is to take advantage of the linear structure of the lists $L_v$ to find a short description of all of the legitimate labelings.  We will show in Section~\ref{sec:faster} how to do this in time $n\cdot\poly( d 2^r / \eps \delta )$ by leveraging the sparsity of $\cC$'s parity-check matrix.

\subsection{Related Work}\label{sec:related}
\paragraph{Work on list-decoding expander codes.}
The work that is perhaps the most related to ours is \cite{HW18}, which seeks to \em list-recover \em expander codes in the presence of erasures in linear time.\footnote{We note that other works, such as \cite{GI04a}, have also had this goal, but to the best of our knowledge \cite{HW18} obtains the best known results, so we focus on that work here.}
List-recovery is a variant of list-decoding which applies to codes over a large alphabet $\Sigma$: instead of receiving as input a vector $z \in \{0,1\}^N$, the decoder receives a collection of lists, $S_1, \ldots, S_N \subseteq \Sigma$, and the goal is to return all codewords $c \in \Sigma^N$ so that $c_i \in S_i$ for all $i$.  In the setting of erasures, 
some lists have size $|\Sigma|$, in which case we may as
well replace the whole list by a $\bot$ symbol.

List decoding from erasures is a special case of list-recovery with erasures,
where the $S_i$  that are not $\bot$  have size one.  However, existing list-recovery algorithms will not immediately work in our setting, as we consider binary codes: list-recovery is only possible for codes with large alphabets. 

Our first observation is that the approach of \cite{HW18} for erasure list-recovery can be used to obtain an algorithm for erasure list-decoding in linear time, even for binary codes.  As described above, our first step is to erasure list-decode $\cC_0$ at each vertex, leaving us with lists $L_v$ that need to be ``stitched together.''  The approach of \cite{HW18} does precisely this, although in their context the lists that they are stitching together come from list-recovering the inner code.  

However, the results of \cite{HW18} about stitching together lists do not immediately yield anything meaningful for erasure list-decoding.  More precisely, 
those results imply that an expander code $\cC(G, \cC_0)$ formed from a graph $G$ with expansion $\lambda$ and an inner code $\cC_0$ with distance $\delta$ and $r$'th generalized distance $\delta_r$ is list-decodable from up to a $\delta \delta_r \inparen{ \frac{ \delta - \lambda/d }{6} }$ fraction of erasures in time $N \cdot \exp( \exp( \exp(r) ) )$.  
In particular, the fraction of erasures that those results tolerate is smaller than the distance of the expander code, yielding only trivial results in this setting. 

Thus, while we use the same ideas as \cite{HW18}, our analysis is different and significantly tighter.  This allows us to obtain a meaningful result in our setting, corresponding to the algorithm \textsc{\LDS}.  Moreover, as described above, we are able to take advantage of the additional linear structure in our setting to improve the dependence on $r$ in the running time.

To the best of our knowledge, there is no algorithmic work on list-decoding expander codes from errors (rather than erasures) in linear time with good parameters. 
We note that \cite{MRRSW19} recently showed that there are expander codes which are \em combinatorially \em near-optimally list-decodable from errors, but this work is non-constructive and does not provide efficient list decoding algorithms.

\paragraph{Work on erasure list-decoding more generally.} It is known that, non-constructively, there are erasure-list-decodable codes of rate $\Omega(\eps)$ which can list-decode up to a $1- \eps$ fraction of erasures, with list sizes $O(\log(1/\eps))$~\cite{Gur03}.  However, this proof is non-constructive and does not provide efficient algorithms, and it has been a major open question to achieve these results efficiently.  Recent progress has been made by \cite{BDT18}, who provided a construction (although no decoding algorithm) with parameters close to this for $\eps$ which is polynomially small in $n$.

Our work is somewhat orthogonal to this line of work on erasure list-decoding
for several reasons.  First, that line of work is mostly concerned with low-rate
codes that are list-decodable from a large fraction of erasures (approaching
$1$), while expander codes tend to perform best at high rates.
Second, we are less concerned with the trade-off between rate and erasure tolerance and more concerned with efficiently erasure-list-decoding an arbitrary expander code as far beyond its (designed) distance as possible.
Finally, much of the line of work described above has focused on getting the list size down to $O(\log(1/\eps))$, which is known to be impossible for linear codes, where the best list size possible is $\Omega(1/\eps)$ \cite{Gur03}.
Since the expander codes we consider are linear, we do not focus on that goal in our work.

\paragraph{Organization.} In Section~\ref{sec:prelim}, we formally introduce the notation and definitions we will need. In Section \ref{sec:slow}, we introduce our preliminary algorithm \textsc{\LDS}, while in Section \ref{sec:faster} we describe the final algorithm that has better dependence on $\epsilon, \delta$ and $r$ in the running time. This proves our Main Theorem \ref{thm:main}.
We conclude in Section \ref{sec:delta2} with the proof of Lemma \ref{lem:delta2}, showing a bound on the second generalized distance of expander codes.

\section{Preliminaries}\label{sec:prelim}

\paragraph{Expander Graphs.}
Let $G = (L \cup R,E)$ be a bipartite graph.\footnote{In this paper we only consider undirected graphs.}
For a vertex $v \in L \cup R$, let $\Gamma(v)$ denote the set of vertices adjacent to $v$.
 For $S \subseteq L$ and $T \subseteq R$, let $E(S,T)$ denote the set of edges with endpoints in $S \cup T$, and for 
$A \subseteq L \cup R$, let $E(A) := E(A \cap L, A \cap R)$.

Let $\tilde{G} = (\tilde{V},\tilde{E})$ be a (not necessarily bipartite) $d$-regular graph on $n$ vertices.
The \em expansion \em of $\tilde G$ is $ \lambda := \max\{ \lambda_2, |\lambda_n| \},$
where $\lambda_1 \geq \lambda_2 \geq \ldots \geq \lambda_n$ are the eigenvalues of the adjacency matrix of $G$.
The  \em double-cover \em  of $\tilde{G}$
is the bipartite graph $G = (L \cup R,E)$ defined as follows.  Let $L$ and $R$ be
two copies of $\tilde{V}$; there is an edge between $u \in L$ and $v \in R$ if and only if $(u,v) \in \tilde{E}$.  
If $\tilde{G}$ is an expander graph, then $G$ obeys the \em Expander Mixing Lemma: \em

\begin{theorem}[Expander Mixing Lemma, see \em e.g. \em \cite{HLW06}] \label{thm:eml}
Suppose that $G = (L \cup R,E)$ is the double cover of a $d$-regular expander graph on $n$ vertices with expansion $\lambda$.
Then for any $S \subseteq L$ and $T \subseteq R$,
$$\left|E(S,T)- \frac{d}{n}|S||T| \right| \leq \lambda \sqrt{ |S||T| }.$$
\end{theorem}

\paragraph{Expander Codes.}
Let $G = (L \cup R,E)$ be the double cover of a $d$-regular expander graph on $n$ vertices, as above.
Let $\cC_0 \subseteq \F_2^d$ be a linear code, called the \em inner code. \em 
Fix an order on the edges incident to each vertex of $G$, and let $\Gamma_i(v)$ denote the $i$'th neighbor of $v$.

The expander code $\cC := \cC(G, \cC_0)$ is defined as the set of all labelings of the edges of $G$ that respect the inner code $\cC_0$.  More precisely, we have the following definition.
\begin{definition}[Expander Code]
Let $\cC_0 \subseteq \F_2^d$ be a linear code, and let $G = (L \cup R,E)$ be the double cover of a $d$-regular expander graph on $n$ vertices.
The expander code $\cC(G,\cC_0) \subseteq \F_2^E$ is a linear code of length $nd$, so that
for $c \in \F_2^{E}$,  $c \in \cC$ if and only if, for all $v \in L \cup R$,
\[ \left(c_{(v,\Gamma_1(v))}, c_{(v, \Gamma_2(v))}, \ldots, c_{(v, \Gamma_d(v))} \right) \in \cC_0. \]
\end{definition}

By counting constraints, it is not hard to see that if $\cC_0 \subseteq \F_2^d$ is a linear code of rate $R$, then $\cC(G, \cC_0) \subseteq \F_2^{E}$ is a linear code of rate at least $2R - 1$.
Moreover, it is known that expander codes have good distance: 
\begin{lemma}[\cite{SS96,Z01}]\label{lem:distance}
Let $\cC_0 \subseteq \F_2^d$ be a linear code with distance $\delta$, and let $G = (L \cup R,E)$ be the double cover of a $d$-regular expander graph with expansion $\lambda$.
Then the expander code $\cC(G, \cC_0)$ has distance at least $\delta( \delta - \lambda/ d)$.
\end{lemma}

Moreover, $\cC$ can be uniquely decoded up to this fraction of erasures in linear time.

\begin{lemma}\label{lem:uniquedec}
Let $\cC_0 \subseteq \F_2^d$ be a linear code with distance $\delta$, and let $G = (L \cup R,E)$ be the double cover of a $d$-regular expander graph on $n$ vertices with expansion $\lambda$.
Let $\epsilon > 0$, and suppose that $\frac \lambda d < \frac \delta 2$. 
Then there is an algorithm \textsc{UniqueDecode} which uniquely decodes the expander code $\cC(G, \cC_0)$ from up to $(1 - \epsilon)\delta(\delta - \lambda/d)$ erasures in time $n \cdot \poly(d)/\epsilon$.
\end{lemma}
The above lemma is by now folklore, but 
for completeness, we include a proof in Appendix~\ref{app:uniquedec}.

\section{A preliminary algorithm}\label{sec:slow}

For clarity of exposition, we begin the proof of our Main Theorem~\ref{thm:main} by proving the following weaker theorem.

\begin{theorem}\label{thm:slow}
Let $\cC_0 \subseteq \F_2^d$ be a linear code with distance $\delta$ and $r$'th generalized distance $\delta_r$.  Let $G = (L \cup R,E)$ be the double cover 
 of a $d$-regular expander graph on $n$ vertices with expansion $\lambda = \max\{\lambda_2, |\lambda_n|\}$.
Let $\cC = \cC(G, \cC_0)$ be the expander code that results.
Let $\eps > 0$, and suppose that
$ \frac{\lambda}{d} \leq \frac{ \eps^2 \delta^2 }{ 2^{r+4} }. $
Let $s:=\frac{ 2^{2r+7}} {\epsilon^4 \delta^4}$.
Then there is an algorithm \textsc{\LDS} which erasure-list-decodes $\cC$ from $( 1- \eps)\delta \delta_r dn$  erasures 
with list size at most $2^{s}$ in time $n \cdot \poly(d) \cdot \exp(s)$.
\end{theorem}

Theorem~\ref{thm:slow} still provides a linear-time algorithm (provided $d, r, \epsilon, \delta$ are all constant), but the dependence on $r$, $\eps$, $\delta$ is not very good.  We will prove Theorem~\ref{thm:slow} in this section to illustrate the main ideas, and then in Section~\ref{sec:faster}, we will show how to adapt the algorithm to achieve the running times advertised in Theorem~\ref{thm:main}.

A formal description of our algorithm \textsc{\LDS} is given in Figure~\ref{fig:alg1}. 
Roughly, the first step
is to list decode the inner codes to obtain an inner list $L_v$ at each vertex $v \in L \cup R$. The second and main step then is to label large equivalence classes by iterating over all possible assignments to representatives from these classes. In the third and final step
 we complete any such possible assignment, by first uniquely decoding at inner codes where sufficient number of edges are already labeled, followed by global unique decoding to recover the rest of the unlabeled edges. Below we elaborate on each of these steps.

In what follows, suppose that $z \in (\F_2 \cup \{\bot\})^{E}$ is a received word with at most $(1 - \eps)\delta \delta_r dn$ symbols that are $\bot$, and let $\mathcal{L} = \mathrm{List}_{\cC}(z)$ be the set of codewords of $\cC$ that are consistent with $z$.

\subsection{List decoding inner codes}\label{subsec:list_dec_inn}

The first step is to list decode all inner codes with not too many erasures.
Specifically, let $B \subseteq L \cup R$ be the set of \em bad \em vertices $v$ so that $z$ has more than $\delta_r d$ erasures incident to $v$.
\begin{equation}\label{eq:defB}
 B = \left\{ v \in L \cup R \suchthat z_{(v,u)} = \bot \text{ for more than $\delta_r d$ vertices $u$} \right\}.
\end{equation}
Then by our assumption on the number of erasures in $z$,
\[ |B \cap L| \delta_r d \leq(1 - \eps)  \delta \delta_r nd \]
and the same for $B \cap R$, which implies that
\begin{equation}\label{eq:Bsmall}
|B \cap L|,|B \cap R| \leq (1 - \eps)\delta n.
\end{equation}

The first step of the algorithm will be to list-decode the inner code $\cC_0$ at every vertex $v \not\in B$.
For all such $v$, let
\begin{equation}\label{eq:defLv}
 L_v := \mathrm{List}_{\cC_0}\left(( z_{(v, \Gamma_1(v))}, z_{(v,\Gamma_2(v))}, \ldots, z_{(v, \Gamma_d(v))} )\right).
\end{equation}

Next we shall use the following notion of \em local equivalence relation \em  to assign labels to many of the edges. 
To define this notion, note first that since $\cC_0$ has $r$'th generalized distance $\delta_r$, 
for any $v \notin B$, $L_v$ is an affine subspace of $\F_2^d$ of dimension $r_v \leq r-1$.
Let $G_v \in \F_2^{d \times r_v}$ and $b_v \in \F_2^d$ be such that
\[ L_v = \inset{ G_v x + b_v \suchthat x \in \F_2^{r_v} }. \]
Notice that each row of $G_v$ corresponds to an edge adjacent to $v$.

Next we define, for any vertex $v \notin B$, a \em local equivalence relation \em  $\sim_v$ at the vertex $v$.

\begin{definition}[Local equivalence relation]\label{def:local_equiv}
Suppose that $v \not\in B$.
For $(u,v), (w,v) \in E$, say that $(u,v) \sim_v (w,v)$ if the row of $G_v$ corresponding to $(u,v)$ is the same as the row of $G_v$ corresponding to $(w,v)$.  
\end{definition}
Notice that Definition~\ref{def:local_equiv} depends on both $v$ and $z$; we suppress the dependence on $z$ in the notation.
We make the following observations.
\begin{observation}\label{obs:ABC}
Suppose that $v \not\in B$.
\ \begin{enumerate}
\item[(A)] If $(u,v) \sim_v (w,v)$, then for any $c
\in \mathcal{L}$, $c_{(u,v)}$ is determined by $c_{(w,v)}$.
\item[(B)] There are  
at most $2^{r-1}$ local equivalence classes at $v$, because there are at most $2^{r-1}$ possible vectors in $\F_2^{r_v}$ that could appear as rows of the matrices $G_v$.  
\end{enumerate}
\end{observation}

\subsection{Labeling large equivalence classes}\label{subsec:label_equiv}

The next step is to assign labels to large \em global \em equivalence classes, defined below.
For this, we first define a new edge set $E' \subseteq E$ by first throwing out all edges touching $B$, and then repeatedly throwing out  edges whose local equivalence classes are too small. Specifically, define $E'$ to be the output of the following Algorithm \textsc{FindHeavyEdges}, given in Figure \ref{algo:fhe}.

\begin{figure}
\centering
\fbox{
\begin{tabular}{cc}
\begin{minipage}{.9\textwidth}
\vspace{.3cm}
\textbf{Algorithm:} \textsc{FindHeavyEdges}

\textbf{Inputs:} A description of $G = (L \cup R,E)$ and $\cC_0 \subseteq \F_2^d$, and the lists $L_v$ for $v \not\in B$.

\textbf{Output:} The set $E' \subseteq E$. 

\bigskip

\textbf{Initialize:} $E' \gets E $.

\begin{enumerate}
\item Remove from $E'$ all edges incident to a vertex in $B$.
\item While true:\\ \\
 If there is some $(u,v) \in  E'$, so that 
 $$\left|\left\{ (w,v) \in E' \suchthat (w,v) \sim_v (u,v) \right\}\right| \leq \frac{ \eps^2 \delta^2  }{ 2^{r+3}} \cdot d,$$
remove $(u,v)$ and all edges $(w,v) \in E'$ so that  $(w,v) \sim_v (u,v)$ from $E' $.
	\item  Break and return the set $E'$.
\end{enumerate}
\end{minipage}
\vspace{.3cm}
&\hspace{.3cm}
\end{tabular}
}
\caption{\textsc{FindHeavyEdges} }
\label{algo:fhe}
\end{figure}

Next we define a \em global \em equivalence relation $\sim$ on the edges in $E'$ as follows.
\begin{definition}[Global equivalence relation]\label{def:global_equiv}
Suppose that $e, e' \in E'$. We say that $e \sim e'$ if there is a path $e=e_1, e_2, \ldots, e_t = e'$
so that $e_1, e_2, \ldots,e_t \in E'$, and for any pair of adjacent edges  $e_i = (u,v)$, $e_{i+1} = (v,w)$ on the path it holds that $(u,v) \sim_v (v,w)$.
\end{definition}

The following lemma shows that  $ E'$ is partitioned into a small number of large global equivalence classes. Consequently, one can assign labels to all edges in $ E'$ by iterating over all possible assignments for a small number of representatives from these classes.

\begin{lemma}\label{lem:equiv_1}
Any global equivalence class in $E'$ has size at least  $\frac{ \eps^4 \delta^4  }{ 2^{2r+7}} dn$. In particular, $E'$ is partitioned into at most $s:=\frac{ 2^{2r+7}} {\epsilon^4 \delta^4}$ different equivalence classes.
\end{lemma}

\begin{proof} 
Let $F$ be a global equivalence class in $E'$, and let $S\subseteq L$ and $T \subseteq R$ denote the left and right 
vertices touching $F$, respectively.
By the definition of $E'$, any vertex $v \in  S \cup T$ is incident to at least $\frac{\eps^2 \delta^2} {2^{r+3}} \cdot d$ edges in $F$. 
Thus by the Expander Mixing Lemma (Theorem \ref{thm:eml}), 
$$
\frac{\eps^2 \delta^2} {2^{r+3}}d\sqrt{|S| |T| }
\leq \left|F\right| 
\leq \frac{d}{n} |S| |T| + \lambda \sqrt{|S| |T| },
$$
and rearranging
\[ n \inparen{ \frac{\eps^2 \delta^2} {2^{r+3}} - \frac{ \lambda }{d} } \leq \sqrt{ |S| |T|}.\]
This implies in turn that 
$$|F| \geq  \frac{\eps^2 \delta^2} {2^{r+3}} d \sqrt{|S| |T| } \geq \frac{\eps^2 \delta^2} {2^{r+3}}  \left( \frac{\eps^2 \delta^2} {2^{r+3}}  - \frac \lambda d\right) dn,$$
which gives the final claim
by our choice of $\frac \lambda d \leq \frac {\epsilon^2 \delta^2} {2^{r+4}}$.
 \end{proof}
 
Finally, by (A) in Observation~\ref{obs:ABC}, choosing a symbol on an edge determines all the symbols in that edge's equivalence class.  Thus, we will exhaust over all choices of symbols for the equivalence classes in $E'$; this leads to $2^s$ possibilities. Next we show that any such choice determines a unique codeword in $\cC$.

\subsection{Completing the assignment}\label{subsec:complete}

To complete the assignment we first show that many of the vertices
have at least $(1-\delta)d$ incident edges in $E'$. For any such vertex, the inner codeword at this vertex is completely determined by the assignment to edges in $E'$, and so can be recovered by uniquely decoding locally at this vertex. We then recover the small number of remaining edges using global unique decoding. Specifically, let 
\begin{equation}\label{eq:defB'}
 B' = \left\{ v \in L \cup R \suchthat (v,u) \notin E'  \text{ for more than $\delta d$ vertices $u$} \right\}.
\end{equation}
The next lemma bounds the size of $B'$, and the number of edges in $E(B')$.

\begin{lemma}\label{lem:equiv_2}
The following hold:
\begin{enumerate}
\item $|B' \cap L|, |B' \cap R| \leq \left(1-\frac \epsilon 2 \right) \delta n$.
\item $|E(B')| \leq \left(1-\frac \epsilon 4\right) \left(\delta - \frac \lambda d\right)\delta nd$.
\end{enumerate}
\end{lemma}

\begin{proof}
For the first item, let $B_1 \subseteq (L \cup R)\setminus B$ be the subset of vertices $ v \notin  B$ so that more than $\left(1- \frac \epsilon 2 \right) \delta d$ edges incident to $v$ are removed on Step 1 of \textsc{FindHeavyEdges}, and let $B_2 \subseteq (L \cup R)\setminus B$ be the subset of vertices $ v \notin B$ so that more than $\frac \epsilon 2 \delta d$ edges incident to $v$ are removed on Step 2 of \textsc{FindHeavyEdges}. 
Note that $B' \subseteq   B \cup B_1 \cup B_2$,
so it suffices to show that $|(B \cup B_1 \cup B_2) \cap L| \leq \left(1-\frac \epsilon 2 \right) \delta n$, and similarly for $R$. 
By (\ref{eq:Bsmall}), $|B \cap L|, |B \cap R| \leq (1-\epsilon) \delta n$. Claims \ref{clm:remove_step_1} and \ref{clm:remove_step_2} below show that 
each of  the sets $B_1, B_2$ has size at most $\frac \epsilon 4 \delta n$ which gives the desired conclusion.

For the second item, note that by the first item and the Expander Mixing Lemma,
\begin{align*}
\left|E(B')\right| & \leq \frac{d}{n} \left( \delta n \left(1 - \frac \eps 2\right) \right)^2 + \lambda \delta n \left(1 - \frac \eps 2\right) \\
&\leq  \left(1 - \frac \eps 2\right)\left( \delta + \frac \lambda d \right)\delta nd \\
& \leq\left(1-\frac \epsilon 4\right) \left(\delta - \frac \lambda d\right)\delta nd,
\end{align*}
where the last inequality follows by our choice of  $\frac{\lambda} {d} \leq \frac{ \epsilon \delta} {8}$. 
\end{proof}

\begin{claim}\label{clm:remove_step_1}
$|B_1| \leq \frac \epsilon 4 \delta n.$
\end{claim}

\begin{proof}
By the description of \textsc{FindHeavyEdges},
$B_1$ is the set of all vertices $v \in (L \cup R) \setminus B$ that are incident to more than $\left(1- \frac \epsilon 2 \right) \delta d$ vertices of $B$.
Thus by the Expander Mixing Lemma,
\begin{align*}
|B_1\cap L|  \left(1- \frac \epsilon 2 \right) \delta  d&\leq \left|E(B_1\cap L, B\cap R)\right| \\
&\leq \frac{d}{n} |B_1\cap L||B\cap R| + \lambda \sqrt{ |B_1\cap L||B\cap R| } \\
&\leq \frac{d}{n} |B_1\cap L| n\delta(1 - \eps) + \lambda \sqrt{ |B_1\cap L| n\delta (1 - \eps) },
\end{align*}
where the last inequality follows by (\ref{eq:Bsmall}).

Rearranging, we have
$$
\sqrt{|B_1\cap L|} \leq \frac{ \lambda \sqrt{ n \delta (1 - \eps) } }{ d\delta \eps/2 },
$$
and
$$
|B_1 \cap L| \leq 4n \inparen{ \frac{\lambda}{d} }^2 \cdot \frac{ 1 }{\delta \eps^2 } \leq  \frac{\epsilon } {8} \delta n,
$$
where the last inequality follows by our choice of $\frac \lambda d \leq \frac{\eps^{3/2} \delta} {8}$.
As the same holds for $B_1\cap R$, we conclude that $B_1$ has size at most $\frac{\epsilon} {4} \delta n$.
\end{proof}

\begin{claim}\label{clm:remove_step_2}
$|B_2| \leq \frac \epsilon 4 \delta n.$
\end{claim}

\begin{proof}
Since there are at most $2^{r-1}$ local equivalence classes at each vertex $v$, the algorithm \textsc{FindHeavyEdges} performs at most
 $2n \cdot 2^{r-1}$ iterations at Step 2. At each such iteration, at most $\frac{\epsilon^2 \delta^2 } {2^{r+3}}\cdot d$ edges are removed, and so the total number of edges removed at Step 2 of  \textsc{FindHeavyEdges} is  $2n \cdot 2^{r-1} \cdot \frac{\epsilon^2 \delta^2 } {2^{r+3}}\cdot d = \frac{\epsilon^2 \delta^2 } {8}\cdot d n$. Finally, by averaging this implies that there are at most $\frac \epsilon 4 \delta n$ vertices $v$ so that more than $
 \frac \epsilon 2 \delta d$ edges incident to $v$ are removed at this step.
\end{proof}

Next observe that for any vertex $v \not\in B'$, the choices for symbols on $E'$ uniquely determine the codeword of $\cC_0$ that belongs at the vertex $v$.  This is because $\cC_0$ has distance $\delta$, and at least $(1 - \delta)d$ edges incident to $v$ have been labeled. Note that since $\cC_0$ is a linear code of length $d$, this unique codeword can be found in time $O(d^3)$ by solving a system of linear equations.
Once this is done, the only edges that do not have labels are those in $E(B')$. By Item (2) of Lemma \ref{lem:equiv_2}, there are at most $\left(1-\frac \epsilon 4\right) \left(\delta - \frac \lambda d\right)\delta nd$ such edges.
By Lemma \ref{lem:uniquedec}, these edges can be recovered using global unique decoding in time $n \cdot \poly(d) /\epsilon$.  In this way, we can recover the entire list $\mathcal{L}$.

\bigskip

The algorithm described above is given as
\textsc{\LDS} in Figure~\ref{fig:alg1}.  This algorithm runs in time $n \cdot \poly(d) \cdot \exp(s)$, 
which proves
Theorem~\ref{thm:slow}.  We will show how to speed it up in Section~\ref{sec:faster}, where we will conclude the proof of Theorem~\ref{thm:main}.

\begin{figure}
\centering
\fbox{
\begin{tabular}{cc}
\begin{minipage}[c][]{.9\textwidth}
\vspace{.3cm}
\textbf{Algorithm:} \textsc{\LDS}

\textbf{Inputs:} A description of $G = (L \cup R,E)$ and $\cC_0 \subseteq \F_2^d$, and $z \in (\F_2 \cup \{\bot\})^E$.

\textbf{Output:} The list $\mathcal{L} = \mathrm{List}_{\cC}(z)$. 

\bigskip

\textbf{Initialize:} $\mathcal{L} = \emptyset$.

\begin{enumerate}
	\item Let $B \subseteq L \cup R$ be as  in \eqref{eq:defB}. 
	For each $v \not\in B$, run $\cC_0$'s erasure list-decoding algorithm to obtain the lists $L_v$ as in \eqref{eq:defLv}.  For each $v$, this entails finding the kernel of a sub-matrix of $G_v$, which can be done in time $O(d^3)$.  Thus, the time for this step is $n \cdot \poly(d)$.  
		\item Run the algorithm \textsc{FindHeavyEdges} given in Figure \ref{algo:fhe} to find the set $E'$, find the partition of $E'$ into $s$ global equivalence classes, and 
choose representative edges $e^{(1)}, \ldots, e^{(s)}$ from each of the equivalence classes. This can be done in time $O(nd)$ using Breadth-First-Search.
	\item For each $y^{(advice)} \in \F_2^{\{e^{(1)}, \ldots, e^{(s)}\}}$: 
	\begin{enumerate}
	\item For each $i \in [s]$, and for each $e \sim e^{(i)}$, define $y_e$ to be the value uniquely determined by $y^{(advice)}_{e^{(i)}}$, as given in Item (A) of Observation~\ref{obs:ABC}.  This can be done in time $O(nd)$, again by Breadth-First-Search.  

	\item Let $B' \subseteq L \cup R$ be as in \eqref{eq:defB'}.  For each $v \not\in B'$, find the unique $y|_{\inset{v}\times\Gamma(v)}$ so that $y_{(v,u)}$ is consistent with existing assignments to $y$, or determine that no such $y$ exists.  As above, this can be done in time $n \cdot \poly(d)$.  If no such $y$ exists for some $v$, continue to the next choice of $y^{(advice)}$.

	\item Use the linear-time erasure unique decoding algorithm \textsc{UniqueDecode} from Lemma \ref{lem:uniquedec} to find a unique $y \in \F_2^E$ that agrees with all the choices of $y$ made so far, or determine that none exists.  If it exists, add $y$ to $\mathcal{L}$. By Lemma \ref{lem:uniquedec}, this can be done in time $n \cdot \poly(d)/\epsilon$.
\end{enumerate}
\item Return $\mathcal{L}$.
\end{enumerate}
\vspace{.3cm}
\end{minipage}
& \hspace{.3cm}
\end{tabular}
}
\caption{\textsc{\LDS}: Returns $\mathrm{List}_{\cC}(z)$ in time $O_{r,\delta, \eps}(n)$.  However, the dependence on $r, \delta, \epsilon$ is not good, and is improved in \textsc{ListDecode}, given in Figure~\ref{fig:alg2}. }
\label{fig:alg1}
\end{figure}

\vspace{.5cm}

\section{Final algorithm}\label{sec:faster}
The algorithm \textsc{\LDS} runs in
time $O_{r, \delta,\epsilon}(n \cdot \poly(d))$, but the constant inside the $O_{r,\delta, \eps}(\cdot)$ is exponential in $\poly\left(\frac{2^r} {\epsilon \delta} \right)$, since there are $s =  \poly\left(\frac{2^r} {\epsilon \delta} \right)$ equivalence classes, and we exhaust over all $2^s$ possible assignments to representatives from these classes.
In this section, we will show how to do significantly better and obtain a running time that depends polynomially on $2^r,  1/\delta, 1/\epsilon$, finishing the proof of Theorem~\ref{thm:main}.  
The basic idea is as follows.  
Instead of exhausting over all possible ways to assign values to the edges $e^{(1)}, \ldots, e^{(s)}$, we will set up and solve a linear system to find a description of the ways to assign these values that will lead to legitimate codewords. Specifically, we prove the following lemma.

\begin{lemma}\label{alg:findcands}
There is an algorithm \textsc{\FC} which, given the state of \textsc{\LDS} at the end of Step 2, runs in time $n \cdot \poly(d,s)$ and returns
$A \in \F_2^{nd \times s}$,
$b \in \F_2^{nd}$, 
$\hat A \in \F_2^{s \times a}$,
and $\hat b \in \F_2^{s}$ so that
\[ \mathrm{List}_{\cC}(z) = \left\{ Ax + b \suchthat x= \hat A\hat x + \hat b\text{ for some $\hat x \in \F_2^{a}$ } \right\}, \]
where $a := \dim(\mathrm{List}_{\cC}(z))$ satisfies $a \leq s$.
\end{lemma}

The above lemma immediately implies Theorem \ref{thm:main}:
We first run Steps 1 and 2 in \textsc{\LDS} in order to find the set $E'$ and its partition into equivalence classes.  As before, this takes time $n \cdot \poly(d)$.  Next, we run \textsc{\FC} in order to find a linear-algebraic description of the list $\mathcal{L}$, which we return.  This second step takes time $n \cdot \poly(s,d)$, for a total running time of $n \cdot \poly(s,d)$.  Plugging in our definition of $s$ proves Theorem~\ref{thm:main}. The formal description of the final algorithm \textsc{ListDecode} is given in Figure \ref{fig:alg2}.
 The rest of this section is devoted to the proof of Lemma \ref{alg:findcands}.

First, note that every value $y_e$ determined by \textsc{\LDS} is some affine function of the labels on $e^{(1)}, \ldots, e^{(s)}$.  That is, there is some matrix $A \in \F_2^{dn \times s}$ and some vector $b \in \F_2^{dn}$ so that the list generated by \textsc{\LDS} is
\[ \inset{ Ax+ b \suchthat x \in \F_2^s }, \]
where $x := y^{(advice)}$.
Our goal in \textsc{\FC} will thus be to find this $A$ and $b$ efficiently, as well as to find a description of the $x$'s so that $Ax + b$ is actually a codeword in $\cC$. 
An overview of the algorithm \textsc{\FC} is given in Figure~\ref{fig:findcands}, and the steps are described below.

\subsection{Finding $A$ and $b$}
The first step of the algorithm will be to find $A$ and $b$.  
To find this efficiently, we will mirror the decoding algorithm in \textsc{\LDS}, except we will do it while keeping the choices of $y^{(advice)}$ as variables.
As we will see below, this can be done in time $n \cdot \poly(s,d)$. 
For this, we shall find a series $\left(A^{(t)},b^{(t)}\right)$ for $t=0,1, \ldots, T$, where $(A,b)=\left(A^{(T)},b^{(T)}\right)$ as follows.

\paragraph{Finding $A^{(0)}$ and $b^{(0)}$.}
First, let $E_0 :=E'$, and let $A^{(0)} \in \F_2^{E_0 \times s}$ and $b^{(0)} \in \F_2^{E_0}$ such that
\[ (A^{(0)} x + b^{(0)})_e =  y_{e} ,
 \]
where $x := y^{(advice)}$, and 
$y_e$ is as in Step (3a) in Algorithm \textsc{\LDS}.
Note that $A^{(0)}$ has rows which are $1$-sparse, and that $A^{(0)}, b^{(0)}$ 
can be created in time $O(nd)$ given the matrices $G_v$ and vectors $b_v$.  Further note that, for any $c \in \cC$, $c|_{E_0} = A^{(0)} c|_{\{e^{(1)}, \ldots, e^{(s)}\}} + b^{(0)}$.

\paragraph{Finding $A^{(1)}$ and $b^{(1)}$.}
Recalling $B'$ from \eqref{eq:defB'}, let $E_1:= \left(E \setminus E(B')\right) \cup E_0$.  Note that for each $e \in E_1 \setminus E_0$, the label on $e$ can be determined in an affine way from the labels on edges in $E_0$.  More precisely, there is some vector $f^{(e)} \in \F_2^{E_0}$ of weight at most $d$ and some $h^{(e)} \in \F_2$ so that for any $c \in \mathcal{C}$, 
$$
c_e = (f^{(e)})^T \cdot c|_{E_0} + h^{(e)},
$$
and moreover $f^{(e)}$ and $h^{(e)}$ can
be found in time $\poly(d)$ by inverting a submatrix of one of the matrices $G_v$.  

Let $F \in \F_2^{(E_1\setminus E_0) \times E_0}$ be the matrix with the $f^{(e)}$ as rows, let $h\in \F_2^{E_1 \setminus E_0 }$ be the vector with entries $h^{(e)}$, and let
\[ A^{(1)} := \begin{bmatrix} \\ A^{(0)} \\ \\ \hline \\ F A^{(0)}\\  \\ \end{bmatrix} \qquad \text{and} \qquad b^{(1)} := \begin{pmatrix} | \\ b^{(0)} \\ | \\ \hline  | \\ F  b^{(0)} +h \\ | \\\end{pmatrix}. \]
Note that $A^{(1)}, b^{(1)}$ can be created in time $n\cdot \poly(d) \cdot s$ given $A^{(0)}$, $b^{(0)}$, $F$, and $h$.
Further note that for any $c \in \cC$, $c|_{E_1} = A^{(1)} c|_{\{e^{(1)}, \ldots, e^{(s)}\}} + b^{(1)}$.

\paragraph{Finding $A^{(t)}$ and $b^{(t)}$ for $t=2, \ldots, T$.}
At this point, by the analysis above (following from Lemma \ref{lem:equiv_2}), we know that there are at most $\left(1 - \frac \epsilon 4\right)\left(\delta - \frac \lambda d\right) \delta nd$ edges which are not in $E_1$.   If we had labels for the edges in $E_1$, 
then by Lemma~\ref{lem:uniquedec} we can use the algorithm \textsc{UniqueDecode} to recover the rest.

The algorithm \textsc{UniqueDecode} is given in Appendix~\ref{app:uniquedec} in Figure~\ref{fig:uniquedec}.
The basic idea is to iteratively decode $\cC_0$ at vertices in $L$, then $R$, then $L$, and so on, to arrive at a unique assignment for all of the edges.
In order to do this with matrices, we will continue as above, creating $A^{(t)}, b^{(t)}$ from $A^{(t-1)},b^{(t-1)}$ for larger $t$ just as we did for $t=1$.  Note that the sets $E_t$ in \textsc{UniqueDecode} play the same role that they do here; $E_t$ represents the set of edges for which a label can be assigned in step $t$.

More precisely,
suppose that at step $t-1$, \textsc{UniqueDecode} has assigned labels to $E_{t-1}$, and 
suppose that we have $A^{(t-1)} \in \F_2^{E_{t-1} \times s}$ and $b^{(t-1)} \in \F_2^{E_{t-1}}$ so that
for any $c \in \cC$, 
$$
c|_{E_{t-1}} = A^{(t-1)} c|_{\{e^{(1)}, \ldots, e^{(s)}\}} + b^{(t-1)}. 
$$
At the next step, \textsc{UniqueDecode} would have assigned labels to edges in $E_{t} \setminus E_{t-1}$.  We note that the total amount of time 
(over all iterations)
to determine the edges in $E_{t}\setminus E_{t-1}$ is the same as in \textsc{UniqueDecode}, which, with the right bookkeeping, is $n \cdot \poly(d)/\epsilon$. 

Then as above, for every $e \in E_t \setminus E_{t-1}$, there is some vector $f^{(e)} \in \F_2^{E_{t-1}}$ of weight at most $d$, and some $h^{(e)} \in \F_2$ so that  for any $c \in \cC$, 
$$
c_e = (f^{(e)})^T \cdot c|_{E_{t-1}} + h^{(e)},
$$
and moreover, these vectors can be found in time $\poly(d)$.
Then, as above, let $F \in \F_2^{(E_{t}\setminus E_{t-1}) \times E_{t-1}}$ be the matrix with the $f^{(e)}$ as rows, let $h\in \F_2^{E_t \setminus E_{t-1}}$ be the vector with entries $h^{(e)}$, and  let
\[ A^{(t)} := \begin{bmatrix} \\ A^{(t-1)} \\ \\ \hline \\ F A^{(t-1)}\\  \\ \end{bmatrix} \qquad \text{and} \qquad b^{(t)} := \begin{pmatrix} | \\ b^{(t-1)} \\ | \\ \hline  | \\ Fb^{(t-1)} + h \\ | \\\end{pmatrix}. \]
As above, $A^{(t)}, b^{(t)}$ can be created in time $|E_t \setminus E_{t-1}| \cdot \poly(d) \cdot s$, and
for any $c \in \cC$, $$c|_{E_t} = A^{(t)} c|_{\{e^{(1)}, \ldots, e^{(s)}\}} + b^{(t)}.$$

We continue this way until $E_T = E$, which happens eventually by Lemma~\ref{lem:uniquedec}.
Then, the amount of work that has been done to compute $A := A^{(T)}$ and $b:= b^{(T)}$ is
\[  n\cdot \poly(d) \cdot s + \sum_{t=2}^{T} |E_t \setminus E_{t-1} | \cdot \poly(d) \cdot s = n \cdot \poly(d,s), \]
as claimed.

\subsection{Finding $\hat A$ and $\hat b$}

Once we have found $A$ and $B$, our goal is to find the set of $x \in \F_2^s$ so that
\begin{equation}\label{eq:wantH}
HAx + Hb= 0,
\end{equation}
where $H \in \F_2^{2nd(1 - R) \times nd}$ is the parity-check matrix for $\cC$.

First, notice that given the parity-check matrix for $\cC_0$, $H_0 \in \F_2^{d(1 - R) \times d}$, and a description of $G$, we can access any entry of $H$ in time $O(1)$: for each vertex $v$, there is a parity check for each row of $H_0$ on the edges incident to $v$.
Next, notice that we can compute $HA$ in time $O(nd^2s)$:  each row of $H$ has at most $d$ nonzeros, so for each of the $O(nd)$ rows of $H$, we take time $O(ds)$ to compute the corresponding row of $HA$.  Similarly we can compute $Hb$ in time $O(nd^2)$.

Our goal then is to find the space of $x$'s which lead to legitimate codewords, which is
\[ \mathcal{W} = \left\{ x \in \F_2^s \suchthat HAx = Hb\right\}. \]
To find a description of $\mathcal{W}$, we first find a basis for the row space of $HA$, which we can do in time $O(nd s^3)$: we iterate through the $O(nd)$ rows of $HA$, and check (in time $O(s^3)$) to see if they are linearly independent from the rows we have already found.  If so, we add the new row to our basis and continue.  Suppose that $t \leq s$ is the dimension of the row space of $HA$, and let $j_1, \ldots, j_t \leq 2nd(1 - R)$ be the indices of the rows in the basis; let $J \in \F_2^{t \times s}$ be the submatrix of $HA$ with these rows.

Let $\hat A \in \F_2^{s \times (s-t)}$ be a matrix so that the columns of $\hat
A$ span $\Ker(J) = \Ker(HA)$.  Note that we can compute such an $\hat A$ in time $\poly(s)$ given $J$.  Next suppose there is some $\hat b \in \F_2^s$ so that 
\begin{equation}\label{eq:wantb}
H A \hat b = Hb.
\end{equation}
 Then 
the space we are after is
\[ \mathcal{W} = \inset{\hat A\hat x + \hat b \suchthat \hat x \in \F_2^{s-t}}. \]

If there is no such $\hat b$, then $\mathcal{L} = \emptyset$ and we should return $\bot$.
If such a $\hat b$ exists, we may compute it by finding a solution to the system
\begin{equation}\label{eq:tosolve}
J \hat b = \left(H b\right)_{j_1, \ldots, j_t}
\end{equation}
 which can be done in time $\poly(s)$.  Indeed, 
suppose that there is some $\hat b$ satisfying \eqref{eq:wantb}.  Then $\hat b$ satisfies \eqref{eq:tosolve}, and for any $b'$ which also satisfies \eqref{eq:tosolve}, $b' \in \hat b + \Ker(J) = \hat b+ \Ker(HA)$, and hence $b'$ satisfies \eqref{eq:wantb} as well.
 Then we check to see if this $\hat b$ satisfies $HA\hat b = H b$, which can be done in time $O(nds)$.  
If so, we return $A,b$ and $\hat A, \hat b$.  If not (or if no $\hat b$ satisfying \eqref{eq:tosolve} exists), then we return $\bot$.

\begin{figure}
\centering
\fbox{
\begin{tabular}{cc}
\begin{minipage}[c][]{.9\textwidth}
\vspace{.3cm}
\textbf{Algorithm:} \textsc{\FC}

\textbf{Inputs:} The state of \textsc{\LDS} after step 2. 

\textbf{Output:} $A \in \F_2^{nd \times s}, b \in \F_2^{nd}, \hat A \in \F_2^{s \times a}, \hat b \in \F_2^s$ so that 
\[ \mathcal{L} = \inset{ Ax + b \suchthat x = \hat A\hat x + \hat b \text{ for some } \hat x \in \F_2^{a} }\]
or returns $\bot$ if such things do not exist. 

\begin{enumerate}
\item Form $A^{(1)}, b^{(1)},$ and find $E_1$ as described in the text in time $n \cdot \poly(d,s)$.  
Let $P_0 \subseteq R, P_1 \subseteq L$ be the sets of vertices incident to an edge in $E \setminus E_1$.
\item For $t = 2,3, \ldots$: 
\begin{enumerate}
\item If $P_{t-1} = \emptyset$, break.
	\item Initialize  $P_t \gets \emptyset$ and $E_t \gets E_{t-1}$.
	\item For each vertex $v \in P_{t-1}$ so that $|(\inset{v} \times \Gamma(v)) \cap E_{t-1}| > (1 - \delta)d$:
	\begin{enumerate}
	\item[$\circ$] Remove $v$ from $P_{t-1}$.
	\item[$\circ$] For any $(v,u) \notin E_{t-1}$, add  $(v,u)$ to $E_t$.	
	\end{enumerate}
	\item For each vertex $v \in P_{t-1}$, for any $(v,u) \notin E_{t-1}$, add $u$ to $P_t$.
	
	\em By Lemma \ref{lem:uniquedec}, the total time (over all iterations) for the above steps is $n \cdot \poly(d)/\epsilon$. \em
	\item Find $A^{(t)}$ and $b^{(t)}$ given $A^{(t-1)}, b^{(t-1)}$ so that for all $c \in \mathcal{C}$, 
\[ c|_{E_t} = A^{(t)} c|_{e^{(1)}, \ldots, e^{(s)}} + b^{(t)}. \]
This can be done in time $|E_t \setminus E_{t-1}| \cdot \poly(s,d)$ as described in the text.
	
\end{enumerate}
\item Let $A = A^{(t)}$ and $b= b^{(t)}$.
\item Compute $HA$ and $Hb$, which can be done in time $O(nd^2s)$.
Let $t \leq s$ be the dimension of the row space of $HA$, and 
find $j_1, \ldots, j_t$ so that the rows of $HA$ indexed by $j_1, \ldots, j_t$ form a basis for the row space of $HA$.  This can be done in time $O(nds^3)$.  Let $J \in \F_2^{t \times s}$ be the submatrix of $HA$ with these rows.
\item Find $\hat A \in \F_2^{s \times (s-t)}$ whose columns are a basis for the kernel of $J$, and find $\hat b \in \F_2^s$ so that $J\hat b = (Hb)|_{j_1, \ldots, j_t}$.  This can be done in time $\poly(s)$. 
\item If $HA\hat b \neq Hb$, return $\bot$.  In this case, $\mathcal{L} = \emptyset$.
\item Otherwise, return $A, b, \hat A, \hat b$.
\end{enumerate}
\end{minipage}
\vspace{.3cm}
&\hspace{.3cm}
\end{tabular}
}
\caption{\textsc{\FC}: prunes the list of advice strings $y^{(advice)}$ in \textsc{\LDS} to a space $\mathcal{L} = \inset{ Lx + \ell \suchthat x \in \F_2^{a} }$ and returns this description.}\label{fig:findcands}
\end{figure}

\begin{figure}
\centering
\fbox{
\begin{tabular}{cc}
\begin{minipage}[c][]{.9\textwidth}
\vspace{.3cm}
\textbf{Algorithm:} \textsc{ListDecode}

\textbf{Inputs:} A description of $G = (L \cup R,E)$ and $\cC_0 \subseteq \F_2^d$, and $z \in (\F_2 \cup \{\bot\})^E$.

\textbf{Output:} A matrix $L \in \F_2^{nd \times a}$ and a vector $\ell \in \F_2^{nd}$ so that
 \[ \mathrm{List}_{\cC}(z) = \inset{ L x + \ell \suchthat x \in \F_2^{a} } \]
for some integer $a$ (which does not depend on $n$),
or else $\bot$ if $\mathrm{List}_{\cC}(z)$ is empty.

\begin{enumerate}
	\item Run Steps 1 and 2 from \textsc{\LDS} (Figure~\ref{fig:alg1}).
	\item Run \textsc{\FC} (Figure~\ref{fig:findcands}).
	\item If \textsc{\FC} returns $\bot$, return $\bot$.
	\item Otherwise, \textsc{\FC} returns $A, b, \hat A, \hat B$.
	\item Compute $L = A  \hat A$ and $\ell = A\hat b +b$.
	\item Return $L, \ell$.
\end{enumerate}
\vspace{.3cm}
\end{minipage}
& \hspace{.3cm}
\end{tabular}
}
\caption{\textsc{ListDecode}: Returns a description of $\mathrm{List}_{\cC}(z)$ in time $n \cdot \poly(d,2^r, 1/\delta, 1/\eps)$.} 
\label{fig:alg2}
\end{figure}

\section{Second generalized distance of expander codes}\label{sec:delta2}

In this section we prove Lemma~\ref{lem:delta2}, restated below.

\restatethm{Lemma \ref{lem:delta2}}{restated}{
Let $\cC_0 \subseteq \F_2^d$ be a linear code with distance $\delta$ and second generalized distance $\delta_2$, and
let $G = (L \cup R,E)$ be the double-cover of a $d$-regular expander graph with expansion $\lambda$.
Let $\epsilon>0$, and suppose that $\frac \lambda d \leq \frac{\delta_2 \delta^2 \epsilon^2} {16}$.
Then the expander code $\cC(G, \cC_0)$ has second generalized distance at least 
$
(1-\epsilon)\cdot \delta \cdot \min\{\delta_2, 2\delta\} .
$
}

We first note the following simple claim which provides an equivalent definition of generalized distance. For a vector $x \in \F_2^N$, let $\supp(x):=\{i \in [n] \mid x_i \neq 0\}$.

\begin{claim}\label{clm:general_dist_equiv}
Let $C \subseteq \F_2^N$ be a linear code.
The $r$'th generalized distance of $C$ is  $$\frac{1} {N} \min_{c_1, c_2,\ldots,c_r}\left| \supp(c_1) \cup \supp(c_2) \cup \cdots \cup \supp(c_r) \right|,$$ where the minimum is taken over all $r$-tuples $c_1, c_2, \ldots, c_r$ of linearly independent codewords in $C$.
\end{claim}

We proceed to the proof of Lemma \ref{lem:delta2}. 
Let $c,c'$ be two distinct  non-zero codewords in $\cC(G, \cC_0)$. Let $F:=\supp(c)$ and $F':=\supp(c')$, by Claim \ref{clm:general_dist_equiv} it suffices to show that 
 $$|F \cup F'| \geq (1-\epsilon) \cdot \delta \cdot \min\{\delta_2, 2\delta\}  dn.$$ Let  $W$ denote the subset of vertices $v\in L \cup R$ which satisfy that $c|_{\Gamma(v)}, c'|_{\Gamma(v)}$ are two distinct non-zero  codewords in $\cC_0$. Let $\epsilon_0:=\frac{\delta^{2} \epsilon} {8}$.
 Below we divide into cases.
 
 \paragraph{Case 1: $|F \cap F'| \leq \epsilon_0 d n$.}

In this case, 
\begin{align*}
 |F \cup F'| & = |F| + |F'| - |F \cap F'| \\
& \geq \left(2\delta(\delta -\lambda/d) - \epsilon_0 \right) dn \\
& \geq (1-\epsilon)2\delta^2 dn,
\end{align*}
where the first inequality follows  since by Lemma \ref{lem:distance}, the code
$\cC(G, \cC_0)$
has relative distance at least $\delta(\delta -\lambda/d)$, and the second inequality follows by choice of $\frac{\lambda} {d} \leq \frac{\delta \epsilon} {2}$ and $\epsilon_0 \leq \delta^2 \epsilon$.

\paragraph{Case 2: $|W| \geq  \epsilon_0^2 n$.} Without loss of generality, we may assume that $|W \cap L| \geq |W \cap R|$, so $|W \cap L| \geq \frac{\epsilon_0^2} {2} n$. We apply the expander mixing lemma with $S_1 :=W \cap L$, and $T_1 \subseteq R$ the set of all right vertices  that are incident to an edge from $F \cup F'$. 

Recall that for any vertex $v \in S_1 $ it holds that $c|_{\Gamma(v)}, c'|_{\Gamma(v)}$ are two distinct non-zero codewords in $\cC_0$, and so 
$$\big|\supp\left(c|_{\Gamma(v)}\right) \cup \supp\left( c'|_{\Gamma(v)}\right)\big| \geq \delta_2 \cdot d.$$ Therefore, any vertex $v \in S_1 $ is incident to at least $\delta_2 d$ edges in $F \cup F'$, and so  $|E(S_1 ,T_1)| \geq \delta_2 d|S_1 |$. 

By the expander mixing lemma, the above implies in turn that
$$\delta_2 d|S_1 | \leq |E(S_1 ,T_1)| \leq  \frac{d} {n} |S_1 ||T_1| +\lambda \sqrt{|S_1 ||T_1|},$$
and rearranging gives
$$|T_1| \geq \left( \delta_2 - \frac{\lambda} {d} \sqrt{\frac {|T_1|} {|S_1 |} } \right)  n \geq \left(\delta_2 - \frac{\lambda} {d}
 \cdot \frac {2} {\epsilon_0} \right) \cdot n ,$$
where the second inequality follows by assumption that $|S_1 | \geq \frac{\epsilon_0^2 } {2}n$.

Finally, note that any vertex $v\in T_1$ has at least $\delta d$ incident edges in $F \cup F'$, and so 
$$|F \cup F'| \geq \delta d |T_1| \geq  \delta \left(\delta_2 - \frac{\lambda} {d}  \cdot \frac {2} {\epsilon_0}\right)  d n \geq \delta \delta_2 (1-\epsilon)d n,$$
where the last inequality follows by choice of $\frac \lambda d \leq  \frac{\delta_2 \delta^2 \epsilon^2} {16} = \frac{\delta_2   \epsilon_0 \epsilon} {2}$.

\paragraph{Case 3: $|F\cap F'| \geq \epsilon_0 d n$ and $|W| \leq  \epsilon_0^2 n$.} 
Under these assumptions, Claims \ref{clm:delta2_intersect} and \ref{clm:delta2_diff} below show that  both the intersection
 $F\cap F'$ and the symmetric difference $F \triangle F'$ are of size at least $(1-\epsilon)\delta^2 dn$. This implies in turn that 
 $$|F \cup F'| = |F \cap F'| + |F \triangle F'| \geq (1-\epsilon) 2 \delta^2dn.$$

\begin{claim}\label{clm:delta2_intersect}
$|F \cap F'| \geq (1-\epsilon)\delta^2dn$.
\end{claim}

\begin{proof}
We apply the expander mixing lemma with $S_2 $ ($T_2$, resp.) being the set of all vertices $v \in L \setminus W$ ($v \in R \setminus W$, resp.) that are incident to an edge from $F \cap F'$. 
Without loss of generality we may assume that $|S_2 | \geq |T_2|$.

Next observe that since $\cC_0$ has relative distance at least $\delta$, any vertex $v \in S_2$ is incident to at least $\delta d$ edges in $F \cup F'$. We claim that these edges are in fact contained in $F \cap F'$; Otherwise, $v$ is incident to some edge from $F \cap F'$ and another edge from $F \triangle F'$ which means that $c|_{\Gamma(v)}, c'|_{\Gamma(v)}$ are two distinct non-zero codewords in $\cC_0$, contradicting the assumption that $v \notin W$. We conclude that any vertex $v \in S_2 $ has at least $\delta d$ incident edges that are incident to either $T_2$ or $W$.

Consequently, we have that 
$$|E(S_2 ,T_2)| \geq \delta d |S_2 | -d |W|  =   \left(\delta- \frac{|W|} {|S_2 |}\right) d|S_2 | \geq  \left(\delta  - \frac{\epsilon_0} {1-\epsilon_0}\right)d|S_2 | ,$$
where the last inequality uses the assumptions that  $|W| \leq \epsilon_0^2 n$ and $|F\cap F'| \geq \epsilon_0 d n$, implying in turn that
$$|S_2 | \geq \frac{|F \cap F'|} {d} -|W| \geq \epsilon_0 (1- \epsilon_0) n .$$
On the other hand, by the expander mixing lemma we have that
$$
 |E(S_2 ,T_2)| \leq \frac{d} {n} |S_2 ||T_2| +\lambda \sqrt{|S_2 ||T_2|},
$$

Combining the above, rearranging,  and recalling our assumption that $|S_2 | \geq |T_2|$, gives 
$$|T_2| \geq \left(\delta - \frac{ \epsilon_0} {1-\epsilon_0} -\frac {\lambda} {d} \right ) n.$$ 
Finally, similarly to the above, we conclude that
$$|F \cap F'| \geq \left(\delta  - \frac{\epsilon_0} {1-\epsilon_0}\right)d|T_2| \geq\left(\delta - \frac{ \epsilon_0} {1-\epsilon_0} -\frac {\lambda} {d} \right )^2 dn \geq \delta^2\left(1-  \eps \right)dn,$$
where the last inequality follows by choice of $\epsilon_0 \leq \frac{\delta\epsilon} {8}$ and $\frac{\lambda} {d} \leq \frac{\delta \epsilon} {4}$.
\end{proof}

\begin{claim}\label{clm:delta2_diff}
$|F \triangle F'| \geq (1-\epsilon)\delta^2dn$.
\end{claim}

\begin{proof}
 Similarly to the previous claim, we apply the expander mixing lemma with $S_3$ ($T_3$, resp.) being the set of all vertices $v \in L \setminus W$ ($v \in R \setminus W$, resp.) that are incident to an edge from $F \triangle F'$, and we may assume that  $|S_3| \geq |T_3|$. 
 
 Once more we observe that  any vertex $v \in S_3$ is incident to at least $\delta d$ edges in $F \triangle F'$, and we conclude that any vertex $v \in S_3$ has at least $\delta d$ incident edges that are incident to either $T_3$ or $W$.
Consequently, as before we have that 
$$|E(S_3,T_3)| \geq    \left(\delta- \frac{|W|} {|S_3 |}\right) d|S_3 | \geq  \left(\delta  - \frac{\epsilon_0^2} {\delta\left(\delta - \frac \lambda d\right)-\epsilon_0^2}\right)d|S_3 | ,$$
where the last inequality uses the assumption that  $|W| \leq \epsilon_0^2 n$, and the fact that $|F \triangle F'| \geq \delta \left(\delta -\frac \lambda d\right)dn$ (since $F \triangle F'$ is the support of the non-zero codeword $c+c' \in \cC(G,\cC_0)$), implying in turn that
$$|S_3| \geq \frac{|F \cap F'|} {d} -|W| \geq \delta\left(\delta - \frac \lambda d\right) n - \epsilon_0^2 n .$$

On the other hand, by the expander mixing lemma we have that
$$
 |E(S_3 ,T_3)| \leq \frac{d} {n} |S_3 ||T_3| +\lambda \sqrt{|S_3 ||T_3|},
$$
and combining with the above, rearranging, and recalling our assumption that $|S_3 | \geq |T_3|$, this gives 
$$|T_3| \geq \left(\delta - \frac{\epsilon_0^2} {\delta\left(\delta -\frac \lambda d\right)-\epsilon_0^2}-\frac {\lambda} {d} \right ) n.$$ 
Finally, similarly to the above, we conclude that
$$|F \cap F'| \geq \left(\delta  -\frac{\epsilon_0^2} {\delta\left(\delta -\frac \lambda d\right)-\epsilon_0^2}\right)d|T_3| \geq\left(\delta - \frac{\epsilon_0^2} {\delta(\delta - \lambda/d)-\epsilon_0^2} -\frac {\lambda} {d} \right )^2 dn \geq \delta^2\left(1-  \eps \right)dn,$$
where the last inequality follows by choice of $\epsilon_0^2 \leq \frac{\delta^{3}\epsilon} {8}$ and $\frac{\lambda} {d} \leq \frac{\delta \epsilon} {4}$.
\end{proof}

\section*{Acknowledgements} 
Most of this work was done while the authors were participating in the Summer Cluster on Error-correcting Codes and High-dimensional Expansion at the Simons Institute for the Theory of Computing at UC Berkeley.
We thank the Simons Institute for the hospitality.  

\bibliographystyle{alpha}
\bibliography{refs}

\newcommand{\etalchar}[1]{$^{#1}$}
\begin{thebibliography}{MRR{\etalchar{+}}19}

\bibitem[AS05]{AS05}
A.~Ashikhmin and V.~Skachek.
\newblock {Decoding of expander codes at rates close to capacity}.
\newblock In {\em Information Theory, 2005. ISIT 2005. Proceedings.
  International Symposium on}, pages 317--321. IEEE, 2005.

\bibitem[BDT18]{BDT18}
Avraham {Ben-Aroya}, Dean Doron, and Amnon {Ta-Shma}.
\newblock Near-optimal erasure list-decodable codes.
\newblock Technical Report TR18-065, Electronic Colloquium on Computational
  Complexity, 2018.

\bibitem[BZ02]{BZ02}
Alexander Barg and Gilles Z{\'e}mor.
\newblock {Error exponents of expander codes}.
\newblock {\em IEEE Transactions on Information Theory}, 48(6):1725--1729, June
  2002.

\bibitem[BZ05]{BZ05}
Alexander Barg and Gilles Z{\'e}mor.
\newblock {Concatenated codes: serial and parallel}.
\newblock {\em IEEE Transactions on Information Theory}, 51(5):1625--1634, May
  2005.

\bibitem[BZ06]{BZ06}
Alexander Barg and Gilles Z{\'e}mor.
\newblock {Distance properties of expander codes}.
\newblock {\em IEEE Transactions on Information Theory}, 52(1):78--90, January
  2006.

\bibitem[DJX14]{DJX14}
Yang Ding, Lingfei Jin, and Chaoping Xing.
\newblock Erasure list-decodable codes from random and algebraic geometry
  codes.
\newblock {\em IEEE Transactions on Information Theory}, 60(7):3889--3894,
  2014.

\bibitem[GI02]{GI02}
Venkatesan Guruswami and Piotr Indyk.
\newblock Near-optimal linear-time codes for unique decoding and new
  list-decodable codes over smaller alphabets.
\newblock In {\em Proceedings of the Thiry-fourth Annual ACM Symposium on
  Theory of Computing}, STOC '02, pages 812--821, New York, NY, USA, 2002. ACM.

\bibitem[GI04]{GI04a}
Venkatesan Guruswami and Piotr Indyk.
\newblock {Linear-Time List Decoding in Error-Free Settings}.
\newblock In Josep D\'{\i}az, Juhani Karhum\"{a}ki, Arto Lepist\"{o}, and
  Donald Sannella, editors, {\em Automata, Languages and Programming}, volume
  3142 of {\em Lecture Notes in Computer Science}, pages 695--707. Springer
  Berlin Heidelberg, 2004.

\bibitem[GK16]{GK14}
Venkatesan Guruswami and Swastik Kopparty.
\newblock Explicit subspace designs.
\newblock {\em Combinatorica}, 36(2):161--185, 2016.

\bibitem[GR06a]{GO06}
Philippe Gaborit and Olivier Ruatta.
\newblock Efficient erasure list-decoding of {R}eed-{M}uller codes.
\newblock In {\em 2006 IEEE International Symposium on Information Theory},
  pages 148--152. IEEE, 2006.

\bibitem[GR06b]{GR06a}
Venkatesan Guruswami and Atri Rudra.
\newblock {Achieving list decoding capacity using folded {R}eed-{S}olomon
  codes}.
\newblock In {\em Allerton '06}, 2006.

\bibitem[GS99]{GS-list-dec}
Venkatesan Guruswami and Madhu Sudan.
\newblock Improved decoding of reed-solomon and algebraic-geometry codes.
\newblock {\em IEEE Trans. Information Theory}, 45(6):1757--1767, 1999.

\bibitem[Gur03]{Gur03}
Venkatesan Guruswami.
\newblock List decoding from erasures: Bounds and code constructions.
\newblock {\em IEEE Transactions on Information Theory}, 49(11):2826--2833,
  2003.

\bibitem[GW17]{GW17}
Venkatesan Guruswami and Carol Wang.
\newblock Deletion codes in the high-noise and high-rate regimes.
\newblock {\em IEEE Transactions on Information Theory}, 63(4):1961--1970,
  2017.

\bibitem[GX12]{GX12}
Venkatesan Guruswami and Chaoping Xing.
\newblock Folded codes from function field towers and improved optimal rate
  list decoding.
\newblock In {\em Proceedings of the 44th annual ACM symposium on Theory of
  computing (STOC)}, pages 339--350. ACM, 2012.

\bibitem[GX13]{GX13}
Venkatesan Guruswami and Chaoping Xing.
\newblock List decoding {R}eed-{S}olomon, algebraic-geometric, and {G}abidulin
  subcodes up to the {S}ingleton bound.
\newblock In {\em Proceedings of the 45th annual ACM symposium on Theory of
  Computing (STOC)}, pages 843--852. ACM, 2013.

\bibitem[HLW06]{HLW06}
Shlomo Hoory, Nati Linial, and Avi Wigderson.
\newblock Expander graphs and their applications.
\newblock {\em Bulletin of AMS}, 43(4):439--561, 2006.

\bibitem[HOW15]{HOW13}
Brett Hemenway, Rafail Ostrovsky, and Mary Wootters.
\newblock Local correctability of expander codes.
\newblock {\em Inf. Comput}, 243:178--190, 2015.

\bibitem[HRW19]{HRW19}
Brett Hemenway, Noga {Ron-Zewi}, and Mary Wootters.
\newblock Local list recovery of high-rate tensor codes and applications.
\newblock {\em SIAM Journal on Computing}, (0):FOCS17--157, 2019.

\bibitem[HW18]{HW18}
Brett Hemenway and Mary Wootters.
\newblock Linear-time list recovery of high-rate expander codes.
\newblock {\em Information and Computation}, 261:202--218, 2018.

\bibitem[Kop15]{Kop15}
Swastik Kopparty.
\newblock List-decoding multiplicity codes.
\newblock {\em Theory of Computing}, 11(5):149--182, 2015.

\bibitem[KRSW18]{KRSW18}
Swastik Kopparty, Noga {Ron-Zewi}, Shubhangi Saraf, and Mary Wootters.
\newblock Improved decoding of folded {R}eed-{S}olomon and multiplicity codes.
\newblock In {\em 2018 IEEE 59th Annual Symposium on Foundations of Computer
  Science (FOCS)}, pages 212--223. IEEE, 2018.

\bibitem[LMSS01]{LMSS01}
Michael~G Luby, Michael Mitzenmacher, Mohammad~Amin Shokrollahi, and Daniel~A
  Spielman.
\newblock Efficient erasure correcting codes.
\newblock {\em IEEE Transactions on Information Theory}, 47(2):569--584, 2001.

\bibitem[MRR{\etalchar{+}}19]{MRRSW19}
Jonathan Mosheiff, Nicolas Resch, Noga {Ron-Zewi}, Shashwat Silas, and Mary
  Wootters.
\newblock {LDPC} codes achieve list decoding capacity.
\newblock {\em arXiv preprint arXiv:1909.06430}, 2019.

\bibitem[PV05]{PV05}
F.~Parvaresh and A.~Vardy.
\newblock {Correcting errors beyond the Guruswami-Sudan radius in polynomial
  time}.
\newblock In {\em Foundations of Computer Science, 2005. FOCS 2005. 46th Annual
  IEEE Symposium on}, pages 285--294, Washington, DC, USA, October 2005. IEEE.

\bibitem[RS06]{RS06}
Ron~M Roth and Vitaly Skachek.
\newblock Improved nearly-{MDS} expander codes.
\newblock {\em IEEE Transactions on Information Theory}, 52(8):3650--3661,
  2006.

\bibitem[Sch00]{Scha00}
Hans~Georg Schaathun.
\newblock The weight hierarchy of product codes.
\newblock {\em IEEE Trans. Information Theory}, 46(7):2648--2651, 2000.

\bibitem[SR03]{SR03}
Vitaly Skachek and Ron~M Roth.
\newblock Generalized minimum distance iterative decoding of expander codes.
\newblock In {\em Proceedings 2003 IEEE Information Theory Workshop (Cat. No.
  03EX674)}, pages 245--248. IEEE, 2003.

\bibitem[SS96]{SS96}
Michael Sipser and Daniel~A. Spielman.
\newblock Expander codes.
\newblock {\em IEEE Transactions on Information Theory}, 42(6):1710--1722,
  1996.

\bibitem[WY93]{WY93}
Victor K.-W. Wei and Kyeongcheol Yang.
\newblock On the generalized {H}amming weights of product codes.
\newblock {\em IEEE Trans. Information Theory}, 39(5):1709--1713, 1993.

\bibitem[Z{\'e}m01]{Z01}
Gilles Z{\'e}mor.
\newblock {On expander codes}.
\newblock {\em IEEE Transactions on Information Theory}, 47(2):835--837,
  February 2001.

\end{thebibliography}

\appendix
\section{Erasure unique decoding of expander codes}\label{app:uniquedec}
In this appendix we prove Lemma~\ref{lem:uniquedec}, which we repeat here:

\restatethm{Lemma \ref{lem:uniquedec}}{restated}{
Let $\cC_0 \subseteq \F_2^d$ be a linear code with distance $\delta$, and let $G = (L \cup R,E)$ be the double cover of a $d$-regular expander graph on $n$ vertices with expansion $\lambda$.
Let $\epsilon > 0$, and suppose that $\frac \lambda d < \frac \delta 2$. 
Then there is an algorithm \textsc{UniqueDecode} which uniquely decodes the expander code $\cC(G, \cC_0)$ from up to $(1 - \epsilon)\delta(\delta - \lambda/d)$ erasures in time $n \cdot \poly(d)/\epsilon$.
}

We note that this lemma is well-known and follows from the techniques of \cite{SS96,Z01}.  However, we include its proof for completeness, because our algorithm \textsc{\FC} mirrors its structure. 

The proof of the lemma follows from the algorithm \textsc{UniqueDecode}, given in Figure~\ref{fig:uniquedec}.

\begin{figure}
\begin{center}

\fbox{
\begin{tabular}{cc}
\begin{minipage}{.9\textwidth}
\vspace{.3cm}
\textbf{Algorithm:} \textsc{UniqueDecode}

\textbf{Inputs:} A description of $G = (L \cup R ,E)$ and $\cC_0 \subseteq \F_2^d$, and  $z \in (\F_2 \cup \{\bot\})^E$.

\textbf{Output:} The unique $c \in \cC(G, \cC_0)$ so that $c$ agrees with $z$ on all un-erased positions.

\bigskip

\textbf{Initialize:}
\begin{itemize}
 \item $E_1 :=\{e \in E| z_e \neq \bot \}$
\item $P_0:=\{ v\in R \mid v \; \text{is incident to an edge}\; e \in E \setminus E_1 \}$
\item $P_1:=\{ v\in L \mid v \; \text{is incident to an edge}\; e \in E \setminus E_1 \}$
\end{itemize}

\bigskip

\item For $t = 2,3, \ldots$:
\begin{enumerate}
	\item If $P_{t-1} = \emptyset$, return the fully labeled codeword.
	\item Initialize $P_t \gets \emptyset$ and $E_t \gets E_{t-1}$.
	\item  For each vertex $v \in P_{t-1}$ so that $|(\inset{v} \times \Gamma(v)) \cap E_{t-1} |> (1 - \delta)d$:
	\begin{enumerate}
	\item Run $\cC_0$'s erasure-correction algorithm  to assign labels to the edges incident to $v$.
	\item Remove $v$ from $P_{t-1}$.
	\item For any $(v,u) \notin E_{t-1}$, add  $(v,u)$ to $E_t$.
		\end{enumerate}
			\item For each vertex $v \in P_{t-1}$,  for any $(v,u) \notin E_{t-1}$, add $u$ to $P_t$.
\end{enumerate}
\vspace{.3cm} 
\end{minipage} & \hspace{.3cm}
\end{tabular} }
\end{center}
\caption{ \textsc{UniqueDecode}: Uniquely decodes an expander code from up to $\delta(\delta - \lambda/d)(1 - \epsilon)$ erasures.}\label{fig:uniquedec}
\end{figure}

To see that \textsc{UniqueDecode} is correct, first notice that on any iteration $t=2,3,\ldots$, the set $E_{t-1}$ is the subset of edges that have already been labeled before this iteration, and $P_{t-2} \cup P_{t-1}$ is the set of vertices touching an edge in $E \setminus E_{t-1}$ that we yet need to decode. The following claim bounds the size of $P_t$, and consequently the number of steps the algorithm runs until it terminates on Step 1.

\begin{claim} The following hold:
\begin{enumerate}
\item For any $t \geq 1$,
$|P_{t+1}| \leq (1-\epsilon)  \left(\delta - \frac \lambda d \right)n$.
\item For any $t \geq 2$, $|P_{t+1}| \leq \left(\frac {1} {1+\epsilon} \right)^2 |P_t|$.
\end{enumerate}
\end{claim}

\begin{proof}

For $t = 1,2, 3,\ldots$, let $B_{t-1} \subseteq P_{t-1}$ be the subset of vertices $v \in P_{t-1}$ that are incident to less than $(1-\delta)d$ edges in $E_{t-1}$. 
Then we have
\begin{equation}\label{eq:unique_dec_B}
P_{t+1} \subseteq B_{t-1} \subseteq P_{t-1},
\end{equation}
as all vertices $v \in P_{t-1} \setminus B_{t-1}$ are removed from $P_{t-1}$ on Step (3b), and consequently will not be present in $P_{t+1}$. 

For the first item, note that $|P_3| \leq |B_1| $ by (\ref{eq:unique_dec_B}), and that $|B_1| \leq \left(\delta - \frac \lambda d\right)(1-\epsilon)n$ since there are at most $( 1 - \epsilon) \delta \left(\delta - \frac \lambda d\right) nd$ erasures to begin with.
Moreover, we have that $|P_3| \geq |P_5| \geq |P_7| \geq \cdots$, and consequently $|P_{t+1}| \leq(1-\epsilon)  \left(\delta - \frac \lambda d \right)n$ for any even $t\geq 1$. Similar reasoning shows that the same  holds for any odd $t \geq 1$.

For the second item, note that by the expander mixing lemma, for $t=2,3,\ldots$,
\[ \delta d |B_{t-1}| \leq |E(B_{t-1}, P_t)| \leq \frac{ d} {n}|B_{t-1}||P_t|+ \lambda \sqrt{ |B_{t-1}||P_t| }, \]
as any vertex $v \in B_{t-1}$ has at least $\delta d$ unlabeled incident edges, and those edges are incident to $P_t$.
Rearranging, we have
\[ |B_{t-1}| \leq  \inparen{ \frac{ \lambda/d }{ \delta - |P_{t}|/n } }^2|P_{t}| \leq  \inparen{ \frac{1}{1 + \epsilon}}^2  |P_t|, \]
where the last inequality follows by assumption that $\frac{\lambda} {d} \leq \frac \delta 2$, and since $\frac{|P_t|} {n} \leq( 1 - \epsilon) (\delta - \lambda/d)$ by the first item.
Finally, by (\ref{eq:unique_dec_B}) this implies in turn that
\[ |P_{t+1}| \leq |B_{t-1}| \leq  \inparen{ \frac{1}{1 + \epsilon}}^2  |P_t|. \]
\end{proof}

Using the above claim we conclude that after $O((\log n)/\epsilon)$ iterations the set $P_{t-1}$ is empty, and so the algorithm terminates.  Moreover, the amount of work done is at most
\[  \poly(d) \cdot \sum_{t=1}^\infty |P_t|  = \poly(d)\cdot n \sum_{t=1}^\infty \inparen{\frac{1}{1 + \epsilon}}^{2t}  =  \poly(d) \cdot \frac{n}{\epsilon} , \]
which proves the lemma.

\end{document}